\documentclass{article}
\usepackage[margin=1in]{geometry}
\usepackage[title]{appendix}
\usepackage{amsthm, amsmath, amsfonts,mathtools, fullpage, graphicx, wrapfig, tikz, pgf, caption, subcaption, bbm}

\usepackage[english]{babel}
\usepackage{natbib}

\setcitestyle{authoryear}

\newtheorem{assumption}{Assumption}

\theoremstyle{definition}
\newtheorem{definition}{Definition}

\newtheorem{lemma}{Lemma}
\newtheorem{proposition}{Proposition}
\newtheorem{condition}{Condition}
\newtheorem{remark}{Remark}

\newcommand\ind{\protect\mathpalette{\protect\independenT}{\perp}}
\def\independenT#1#2{\mathrel{\rlap{$#1#2$}\mkern2mu{#1#2}}}
\newcommand{\dsep}{\perp\!\!\!\!\perp_d} 
\newcommand{\E}{\mathbb{E}}

\newcommand{\1}{\mathbbm{1}}

\usetikzlibrary{arrows,shapes.arrows,shapes.geometric,
shapes.multipart,decorations.pathmorphing,positioning,
swigs}
\tikzstyle{condition} = [rectangle, text centered, draw=gray]
	\tikzstyle{from} = [<-, shorten <=1pt, >=stealth', semithick]
	\tikzstyle{into} = [->,shorten <=1pt, >=stealth', semithick]
	\tikzstyle{double} = [<->,shorten <=1pt, >=stealth', semithick]
	\tikzstyle{u} = [shape=circle, draw] %latent variables
	\tikzstyle{e} = [blue] %exposure
	\tikzstyle{o} = [red]  %outcome
	\tikzstyle{i} = [xshift=5mm] %intervention nodes

\title{Using causal diagrams to assess parallel trends in difference-in-differences studies}

\author{Audrey Renson$^1$, Oliver Dukes$^2$ and 
Zach Shahn$^3$ \\
\small $^1$ Department of Population Health, New York University Grossman School of Medicine, New York, NY, USA\\
\small $^2$ Department of Mathematics, Computer Science and Statistics, Ghent University, Ghent, Belgium\\
\small $^3$ School of Public Health, The City University of New York, New York,  NY,  USA}

\usepackage{setspace}
\begin{document}

\maketitle

\begin{abstract}
Difference-in-differences (DID) is popular because it can allow for unmeasured confounding when the key assumption of parallel trends holds. However, there exists relatively little guidance on how to decide a priori whether this assumption is reasonable. We attempt to develop such guidance by considering the relationship between a causal diagram and the parallel trends assumption. This is challenging because parallel trends is scale-dependent and causal diagrams are generally nonparametric (i.e., scale-independent). Here, we develop conditions under which, given a nonparametric causal diagram, one can reject or fail to reject parallel trends. In particular, we adopt a linear faithfulness assumption, which states that all graphically connected variables are correlated, and which is often reasonable in practice. Then, we show that parallel trends can be rejected if either (i) the treatment is affected by pre-treatment outcomes, or (ii) there exist unmeasured confounders for the effect of treatment on pre-treatment outcomes that are not confounders for the post-treatment outcome, or vice versa (or more precisely, the two outcomes possess distinct minimally sufficient sets). We also argue that parallel trends should be strongly questioned if (iii) the pre-treatment outcomes causally affect the post-treatment outcomes (though the two can be correlated) since there exist reasonable semiparametric models in which such an effect violates parallel trends. When (i-iii) are absent, a necessary and sufficient condition for parallel trends is that the association between the common set of confounders and the potential outcomes is constant on an additive scale, pre- and post-treatment. These conditions are similar to, but more general than, those previously derived in linear structural equations models. We discuss our approach in the context of the effect of Medicaid expansion under the U.S. Affordable Care Act on health insurance coverage rates.
\end{abstract}
\section{Introduction}
Difference-in-differences (DID) is a common identification approach, whose popularity likely comes from the understanding that it can allow for unmeasured confounding under certain assumptions. In particular, the key assumption of parallel trends requires that average differences in untreated potential outcomes over time are constant across levels of a binary treatment. Usually, the parallel trends assumption is assessed through examining whether trends are parallel in periods prior to the treatment \citep{callaway2023difference, cunningham2021causal}. Examination of these `pre-trends' can be useful, but has important limitations. Specifically, parallel pre-trends is neither necessary nor sufficient for parallel trends \citep{lechner2011estimation}, and assessing pre-trends requires access to historical data. While some common reasons for the failure of parallel trends have been discussed \cite[see e.g., ][]{ashenfelter1978estimating}, guidance is generally lacking regarding what features of a theoretical data-generating model (DGM) would lead parallel trends to hold or not.

Historically, in the literature, DID has been motivated based on a linear structural equation model \citep[e.g.,][]{angrist2009mostly}, and in this case the features of the DGM that imply parallel trends are more clear. Often a proposed model takes a form similar to 
\begin{align}\label{eq:twfe}
    Y_{it}^0&=\alpha_i + \lambda_t +\varepsilon_{it},
\end{align}
\citep{callaway2023difference}, where $Y_{it}^0$ denotes a potential outcome under no treatment, $\alpha_i$ are unit-specific ``fixed effects'' that can be conceptualized as time-invariant confounding factors \citep{sofer2016negative}, $\lambda_t$ captures common time trends ($t=0,1$), and $\varepsilon_{it}$ are independent errors. (See also \cite{wooldridge2021two} and \cite{angrist2009mostly} for similar expositions.) The fact that the magnitude of $\alpha_i$ (on an additive scale) does not depend on time is key in guaranteeing parallel trends under this model. However, it is not clear (a) how one would use a priori scientific information to justify the parametric form of model (\ref{eq:twfe}), nor is it clear (b) to what extent the strong conditions contained in this model are necessary. It would be desirable to be able to reason about parallel trends as presented in more recent expositions of DID that do not a  priori restrict the functional form of any models \citep{abadie2005semiparametric, lechner2011estimation, callaway2023difference}.

Causal graphs have proven useful for assessing identification assumptions in other causal inference approaches (e.g., instrumental variables, unconfoundedness) \citep{shrier2008reducing}. In particular, graphs can translate background scientific information in the form of the presence or absence of simple cause and effect relationships, which are often known, into the more obtuse statistical properties of the DGM needed for identification (such as conditional exchangeability/ignorability). This functionality would be valuable in DID applications, but is currently lacking. The challenge in applying graphs to DID is that parallel trends, unlike conditional exchangeability, represents a scale-dependent restriction on the association between potential outcomes and treatment \citep{lechner2011estimation}. For this reason, parallel trends cannot be rejected or accepted based on a nonparametric graph alone (except in trivial cases). Thus, it remains unclear, outside of special cases involving fully linear DGMs \citep{kim2021gain}, how one may translate simple cause and effect relationships to a statement about whether parallel trends should hold. Put another way, the unmeasured confounding that DID can adjust for has not been formally connected to the graphical concept of confounding, limiting the ability to draw on the extensive literature regarding the latter.

\par In this paper we explore the relationship between parallel trends and the information that can be encoded in a causal diagram, in order to develop graphical (and necessarily extra-graphical) conditions to guide adoption of the parallel trends assumption. We focus on the canonical two-group (treated, untreated), two-period (pre- and post-treatment) setting without measured covariates, and begin with a nonparametric structural equation model that allows common causes of any two or all three variables observed under this setup. We find that, while the parallel trends assumption is neither implied nor negated by a graph without additional assumptions, several useful necessary conditions on a graph can be derived if one adopts a (linear) faithfulness condition. Linear faithfulness states that any two variables which are not d-separated must be correlated, and is usually plausible in practice. Under linear faithfulness, parallel trends is negated by a graph encoding either (i) an effect of pre-treatment outcomes on treatments, or (ii) disjoint sets of unmeasured confounding variables for the treatment effect on pre-treatment and post-treatment outcomes (which we formalize using minimally sufficient sets). Additionally, under linear faithfulness and certain semiparametric restrictions on the DGM (though not in general), (iii) an effect of pre-treatment on post-treatment outcomes is only consistent with parallel trends if certain associations coincidentally cancel, and should thus be cause for concern. We also provide a necessary and sufficient condition for parallel trends (supposing our graph is not in conflict with parallel trends). This amounts to an additive homogeneous confounding assumption similar to the one encoded in model (\ref{eq:twfe}), but expressed in terms of an unmeasured covariate set that would identify the treatment effect on pre- and post-treatment outcomes (where the former is always assumed to be zero). 

The remainder of this paper is organized as follows: Section 2 places our analysis in the context of previous literature; Section 3 formalizes several key technical concepts, the assumed data-generating model, and assumptions; Section 4 presents results in the form of necessary and sufficient conditions for parallel trends; Section 5 applies our analysis to the effect of the US Affordable Care Act Medicaid expansions on insurance rates; and Section 6 concludes.
\section{Previous literature}
Our work builds directly on \cite{ghanem2022selection}, who established related necessary and sufficient conditions for parallel trends to hold under a nonparametric model for the potential outcomes using structural equations, but did not relate these to causal diagrams or graphical separation rules. Aided by causal diagrams, we introduce different but related criteria for reasoning about parallel trends, some (but not all) of which are direct implications of Ghanem et al.'s. A key difference between \cite{ghanem2022selection} and our work is our explicit framing around confounding, drawing connections between the DID literature and the broader causal inference literature based on unconfoundedness assumptions. On the other hand, \cite{ghanem2022selection} consider the role of measured covariates and provide sensitivity analyses, which we do not. 

Aside from \cite{ghanem2022selection}, nearly all literature to our knowledge that considers structural conditions for parallel trends does so in the context of linear structural models. As the one other exception, to our knowledge, \cite{weber2015assumption} begin with a nonparametric structural model representing a specific substantive example. These authors then develop a few different implications of parallel trends that include condition (i) in our results, and otherwise involve strong linearity assumptions for certain variables.

In terms of the remaining literature that studies parallel trends under linear structural models, \cite{kim2021gain} is perhaps most similar to our own work. \cite{kim2021gain} consider a few variations on a linear data-generating model and illustrate conditions on the linear coefficients required for parallel trends to hold. In addition to highlighting the more well-known ``constant bias'' condition (similar to our ``additive homogeneous confounding''), they show that if the pre-treatment outcome affects either the treatment or the post-treatment outcome, certain coefficients must cancel exactly in ways that might be hard to justify in practice, similar to conditions (i) and (iii) in our results. Our work builds upon \cite{kim2021gain} by showing that similar conditions are arrived at when one considers an unrestricted, nonparametric data-generating model. 

Additionally,  \cite{zeldow2021confounding} explore conditions under which covariates (observed or not) may cause bias in parametric DID settings where a two-way fixed effects model is appropriate. They arrive at a definition of confounding whose absence is similar to our ``additive homogeneous confounding'' assumption. Thus, our contribution relative to these earlier works assuming a particular parametric model is to illustrate that much of the same intuition holds when extended to a nonparametric model.

\section{Methods}\label{sec1}
\subsection{Difference-in-differences}
We focus on the canonical, two-period, two-group DID scenario, as follows. Consider an outcome variable $Y_{it}$ measured at times $t=0$ and $t=1$ on a sample of $n$ individuals ($i=1,...,n$). Some proportion of individuals are exposed to a binary treatment $A$ ($A=1$ treated, $A=0$ untreated) that occurs between the two outcome measurements, so that all individuals are untreated at time $t=0$. We use $Y_{it}^a$ to denote a potential outcome, or the value that $Y_{it}$ would take if, possibly countrary to fact, $A$ were set by intervention to $a$. Upper case is used to denote random variables, lower case for specific realizations, and script for support. The symbol $\ind$ is used to denote statistical independence. For any random variable $X$, we let $\dot X=X-\E[X]$ denote the variable centered about its mean.

Throughout, we will use the following definition of parallel trends:
\begin{definition} [Parallel trends]  \label{def:pt} $\E(Y_1^0-Y_0^0|A)=\E(Y_1^0-Y_0^0)$
\end{definition}
It is well-known that under parallel trends $\{$along with no anticipation ($Y_0^0=Y_0$) and causal consistency ($A=0\implies Y_1^0=Y_1$)$\}$, the average treatment effect in the treated is identified as $\E(Y_1^1-Y_1^0|A=1)=\E(Y_1-Y_0|A=1)-\E(Y_1-Y_0|A=0)$, and simple estimators based on sample moments are available. 
\subsection{Causal diagrams}
We draw heavily on causal diagrams, so here we provide semi-formal definitions of key concepts. For any random variable $V_j$, let $pa_j\equiv pa(V_j)$ denote the parents of $V_j$; i.e. variables that directly cause $V_j$. Let $f_j$ denote the structural function for $V_j$; i.e. the mechanism by which Nature assigns values of $V_j$ given its parents $pa_j.$ That is, $V_j=f_j(pa_j, \varepsilon_j),$ where $\varepsilon_j$ denotes an independent error term. A causal directed acyclic graph (DAG) is a collection of nodes $V_j$ $\{j=1,...,J\}$ representing random variables, with a directed edge from $V_i$ to $V_j$ if $V_i \in pa_j$. Two nodes are \textit{adjacent} if there is a directed edge between them, and a \textit{path} from $V_i$ to $V_j$ is a sequence of distinct nodes that are adjacent. A \textit{backdoor path} is a path between a treatment $A$ and an outcome that ends in a directed edge pointing into $A$. We use $\dsep$ to denote graphical \textit{d}-separation (for a definition of \textit{d}-separation, see \cite{pearl2009causality}). Unless otherwise stated, all causal diagrams in this paper correspond to nonparametric structural equation models (NPSEMs), such that the structural functions $\{f_j, j=1,...,J\}$ are undefined other than the specification of their parents. 
\par Throughout, we adopt the causal Markov assumption that any node is independent of its non-descendents given its parents, implying that if two variables $V_i$ and $V_j$ are \textit{d}-separated given a third $V_k$, then $V_i \ind V_j | V_k$ \citep{pearl2009causality}.
\par Potential outcomes can be connected to a NPSEM by depicting the latter as a directed single world intervention graph (SWIG) \citep{richardson2013single}. A directed SWIG is formed from a causal DAG by splitting treatment $A$ nodes into observed ($A$) and intervened ($a$) values, such that all edges entering $A$ are retained, all edges emanating from $A$ are deleted and shown emanating from $a$, and all variables $V$ that are descendents of $a$ are replaced with potential outcomes $V^a$. Directed SWIGs obey the causal Markov assumption, so that counterfactual independencies such as $Y^a \ind A|Z$ can be deduced from \textit{d}-separation. We also use partially directed SWIGs, in which directed edges represent known causal relationships as in a directed SWIG, and undirected edges represent open paths whose causal direction is unknown. That is, an undirected edge between two nodes $V_i$ and $V_j$ can represent either (i) unmeasured common causes of $V_i$ and $V_j$, (ii) direct causation of $V_i$ by $V_j$, or (iii) direct causation of $V_j$ by $V_i$; without specifying which. In other words, we use a partially directed SWIG as a shorthand to represent all possible directed SWIGs consistent with the unknown relationships depicted by the undirected edges.

\subsection{Linear faithfulness}
While the causal Markov assumption specifies that \textit{d}-separated variables must be independent, it does not require that \textit{d}-connected variables be associated. In what follows, some of our results will depend on the latter requirement, which we believe to be reasonable in most situations.
\begin{assumption}[Linear faithfulness]\label{asn:f} For any two vertices $V_i$ and $V_j$ in graph $G$, for any subset $B$ of the vertices in $G$, $Cov(V_i, V_j | B)=0 \; a.s.$ implies $V_i \dsep V_j | B$.
\end{assumption}
In words, linear faithfulness states that any variables with zero (conditional) covariance must also be (conditionally) \textit{d-}separated; i.e., any \textit{d-}connected variables cannot have zero covariance. Assumption \ref{asn:f} is called ``linear'' faithfulness, not because it restricts non-linear relationships between variables (it does not), but rather because it requires that \textit{d}-connected variables have an association that is non-null when measured on a linear scale (i.e., via covariance). It is more common to express faithfulness in terms of independence rather than zero covariance, but covariance is what matters for parallel trends, and one might reasonably expect non-zero covariance whenever independence is violated. Indeed, linear faithfulness is often adopted in the causal discovery literature  (see \cite{spirtes2001causation}, pg. 47), where standard algorthms for structure learning rely on (partial) correlations between variables. Violations of Assumption \ref{asn:f} occur when the paths between $V_i$ to $V_j$ cancel perfectly (possibly due to equilibrium forces or external control). In an example given by \cite{steel2006homogeneity}, a city may improve roads, leading to faster driving and increased fatalities, while simultaneously hiring more police officers to enforce speed limits with the goal of maintaining fatalities at existing levels; as a result, there may be no marginal covariance between road quality and fatalities, even though these road quality and fatalities would be d-connected in a graph. However, in the absence of such external constraints, linear faithfulness is usually thought to be reasonable. For additional examples and discussion, see \cite{andersen2013expect}. 

\subsection{Sufficient sets and  minimally sufficient sets}
For treatment $A$ and outcome $Y_t$ ($t=0,1$), we refer to a set of variables (measured or unmeasured) $M$ as a sufficient adjustment set (or simply, sufficient set) for the effect of $A$ on $Y_t$ if $\E(Y_t^0|A, M)=\E(Y_t^0|M)$ $a.s$. Since we are only considering one treatment $A$, we sometimes abbreviate this statement to say that $M$ is a sufficient adjustment set for $Y$. Also, since we are only concerned with identifying the ATT, we only consider untreated potential outcomes $Y_t^0$; in other applications, one may wish to consider $Y_t^a$ for all $a$. We say that a set of variables $M$ is a \textit{minimally} sufficient set if $M$ is a sufficient set and no proper subset of $M$ is sufficient \citep{greenland1999causal}. One can deduce whether $M$ is a sufficient set by checking if $Y_t^0$ is d-separated from $A$ given $M$ in a SWIG under the intervention $a=0$ \citep{richardson2013single}. 

\subsection{Data-generating model}
DID is usually motivated by the existence of unmeasured confounding. Thus, in addition to the observed variables $A$, $Y_0$, and $Y_1$, we will consider unmeasured common causes of any two or all three variables: let $U_1$ denote all variables that directly cause all three observed variables, let $U_2$ denote all variables that directly cause $Y_0$ and $Y_1$ but not $A$; let $U_3$ denote all variables that directly cause $A$ and $Y_0$ but not $Y_1$; and let $U_4$ denote all variables that directly cause $A$ and $Y_1$ but not $Y_0$. We assume that $U_1$, $U_2$, and $U_3$ can impact $U_4$ but leave the relationships between unmeasured variables otherwise unspecified. These relationships are depicted as a partially directed SWIG and corresponding NPSEM in Figure \ref{fig:npsem}. Thus, the model is essentially unrestricted given the time ordering $\left \{ \{U_1, U_2, U_3\} \rightarrow Y_0 \rightarrow U_4 \rightarrow A \rightarrow Y_1 \right\}$ (i.e, the time ordering between $U_1, U_2$, and $U_3$ is left unspecified). The only restrictions on the model are intended to allow for unmeasured variables to have different relationships with each of the measured ones. Specifically, $U_2$ does not directly affect $A$, $U_3$ does not directly affect $Y_1$, and $U_4$ does not directly affect $Y_0$; such effects would make each of these unmeasured variables equivalent to $U_1$ in the graph and thus unnecessary to define.
\begin{figure}
    \centering
    \begin{minipage}[c]{4cm}
      \begin{tikzpicture}[node distance=2cm]
\tikzset{line width=0.7pt, outer sep=0pt,
         ell/.style={draw,fill=white, inner sep=2pt,
          line width=0.7pt}}
        \node (y0) [ell, shape=ellipse] {$Y_0$};
        \node (a) [ell, shape=swig vsplit, right of=y0]  {\nodepart{left}{$A$}
        \nodepart{right}{a=0}};
        \node (y1) [ell, shape=ellipse, right of=a] {$Y_1^0$};
        \node (u1) [ell, dashed, shape=ellipse, below of=a] {$U_1$};
        \node (u2) [ell, dashed, shape=ellipse, above of=a] {$U_2$} ;
        \node (u3) [ell, dashed, shape=ellipse, below of=y0] {$U_3$};
        \node (u4) [ell, dashed, shape=ellipse, below of=y1] {$U_4$};
        \draw[->] (y0) to [bend left=23] (y1);
        \draw[->] (y0) to (a);
        \draw[->] (a) to (y1);
        \draw[->] (u1) to[bend left=10] (y0);
        \draw[->] (u1) to[out=110, in=-155] (a);
        \draw[->] (u1) to[bend right=10] (y1);
        \draw[->] (u2) to[bend right =10] (y0);
        \draw[->] (u2) to (y1);
        \draw[->] (u3) to (y0);
        \draw[->] (u3) to[in=-165] (a);  
        \draw[->] (u4) to (y1);
        \draw[->] (u4) to[out=140, in=-130] (a);
        \draw[->] (u1) to (u4);
        \draw[->] (y0) to[bend right=10] (u4);
        \draw[dashed, -] (u1) to [bend left=45] (u2);
        \draw[dashed,-] (u1) to (u3);
        \draw[dashed,-] (u1) to (u4);
        \draw[dashed,-] (u2) to[bend right=75] (u3);
        \draw[->] (u2) to[bend left=75] (u4);
        \draw[->] (u3) to [bend right=23] (u4);
        \end{tikzpicture}
   \end{minipage}%
   \hfill
   \begin{minipage}[c]{\textwidth-7cm}
      \begin{align*}
\{U_1, U_2, U_3\} &\sim P(U_1, U_2, U_3) \\
Y_0 &= f_{Y_0}(U_1, U_2, U_3, \varepsilon_{Y_0})\\
U_4 &= f_{U_4}(U_1, U_2, U_3, Y_0, \varepsilon_{U_4})\\
A &= f_A(U_1, U_3, U_4, Y_0, \varepsilon_A)\\
    Y_1 &= f_{Y_1}(U_1, U_2, Y_0, U_4, A, \varepsilon_{Y_1})
\end{align*}
   \end{minipage}
    \caption{Partially directed SWIG (left) and corresponding structural equation model (right) assumed to generate data for a cannonical 2x2 DID. Undirected edges (depicted as dashed lines for visual clarity) are used to represent edges whose direction is unknown. (https://dagitty.net/dags.html?id=jqVhYree)}
    \label{fig:npsem}
\end{figure}

\subsection{Motivating example}
To ground ideas, consider the following example. Under the Affordable Care Act (ACA), US states were given the option to expand Medicaid coverage using heavy federal subsidies starting in 2014. Some states took the opportunity to expand coverage right away, while others delayed. Let $i\in\{1,...,50\}$ index the US states, let $A_i$ denote an indicator of state $i$ having expanded Medicaid in 2014, and let $Y_{i0}$ and $Y_{i1}$ denote the uninsurance rate in state $i$ during 2013 and 2014, respectively. Suppose we wish to estimate $\E(Y_1^1-Y_1^0|A=1)$, the average immediate effect of Medicaid expansion in 2014 on health uninsurance rates in 2014 among people living below 138\% of the federal poverty line in expanding states. Parallel trends (Definition \ref{def:pt}) here states that in absence of Medicaid expansion, average differences in insurance rates between 2013 and 2014 would have been equal in expanding and non-expanding states.

\subsection{Prior results on the relationship between parallel trends and confounding}
In this paper, we define confounding at time $t$ as the covariance between the treatment and time $t$ potential outcomes, $Cov(A, Y_t^0)$, noting that this quantity equals zero if marginal exchangeability ($Y_t^0 \ind A)$ holds. An alternative (but not equivalent) way of defining confounding is the difference between the marginal association of $A$ and $Y_t$, $\E(Y_t|A=1)-\E(Y_t|A=0)$, and the true ATT at time $t$, $\E(Y^1_t-Y^0_t|A=1)$. This difference can be shown to equal $\E(Y^0_t|A=1)-\E(Y^0_t|A=0)$. A well known result in the literature is that parallel trends implies that the time 1 and time 0 confounding are equal in magnitude \citep{lechner2011estimation,sofer2016negative,zhang2021exploiting}; this applies whichever definition of confounding we use; i.e.,
\begin{align}
\nonumber    \E(Y_1^0-Y_0^0|A=1) &= \E(Y_1^0-Y_0^0|A=0)  \\
 \implies \E(Y_1^0|A=1)-\E(Y_1^0|A=0) &= \E(Y_0^0|A=1)-\E(Y_0^0|A=0)
 \label{eq:ec1}
    \end{align}
where (\ref{eq:ec1}) corresponds to the second definition, but also
\begin{align}
\nonumber    \E(Y_1^0-Y_0^0|A=1) &= \E(Y_1^0-Y_0^0) \\
%   \nonumber \implies \E[Y_1^0|A=1]-\E[Y_1^0] &= \E[Y_0^0|A=1]-\E[Y_0^0] \\
\nonumber   \implies \E(Y_1^0A)/\E(A)-\E(Y_1^0) &= \E(Y_0^0A)/\E(A)-\E(Y_0^0) \\
 \implies Cov(Y_1^0,A) &= Cov(Y_0^0,A) \label{eq:ec2}
\end{align}
where (\ref{eq:ec2}) corresponds to the first definition. 

Importantly for our work, \cite{ghanem2022selection} recently showed that a more granular statement can be made about the relationship between potential outcomes and treatment under parallel trends. We draw heavily on this result and so restate it here for clarity. 
\begin{lemma}[Ghanem et al., 2024, Lemma F.3]\label{lemma:ghanemE3} Under mild regularity conditions (described in the original Lemma), parallel trends implies 
\begin{equation}\label{eq:ghanemE3}
\E(Y_1^0  -Y_0^0|pa(A))=\E(Y_1^0-Y_0^0)\quad a.s.    
\end{equation}
\end{lemma}
In words, Lemma \ref{lemma:ghanemE3} states that parallel trends over values of the treatment imply parallel trends over values of the \textit{parents} of treatment as well. An equivalent formulation of (\ref{eq:ghanemE3}) is $\E\{\dot Y_1^0|pa(A)\}=\E\{\dot Y_0^0|pa(A)\}$; i.e., parallel trends implies a constant expected deviation of $Y_1^0$ and $Y_0^0$ from their respective means, given the parents of $A$.  In other words, (\ref{eq:ghanemE3}) implies a common functional form for the dependence of $Y_t^0$ on $pa(A)$ across $t$, at the mean. Though the result is stated nonparametrically, intuition can be gained by considering the following parametric model:
\begin{align}\label{eq:commonmodel}
    Y_t^0=\alpha_t +pa(A)'\beta_t + \epsilon_t, \quad t=0,1
\end{align}
where $\beta_t$ is a vector of the same length as $pa(A)$. From Lemma \ref{lemma:ghanemE3}, parallel trends implies that the above regression would have $\beta_0=\beta_1$; i.e., the regression would have the same coefficients at times 0 and 1 up to a possibly time-varying intercept.

\section{Results}\label{sec:implications}
In this section, our aim is to show what features of a SWIG may support or contradict parallel trends in applications. Throughout, we will always assume no anticipation and linear faithfulness. Then, our general approach is to develop conditions under which, if we additionally adopt parallel trends, a SWIG can be rejected. If parallel trends implies that a SWIG should be rejected, this implies (by contrapositive) that parallel trends cannot hold under this SWIG. First, subsection \ref{sec:sufficient} develops some preliminary results extending previous literature on the relationship between parallel trends and the graphical notion of confounding. Then, subsections \ref{sec:y0a} and \ref{sec:distinct} develop the first two conditions which are necessarily implied by parallel trends. Subsection \ref{sec:y0y1} develops a third condition which is not implied (or at least, we are not able to prove is implied) by parallel trends in all models, but which is implied in some reasonable semiparametric models. Then, subsection \ref{sec:summary} summarizes interdependencies between the three conditions.
\subsection{Preliminaries: parallel trends and sufficient adjustment sets}\label{sec:sufficient}
First, we will show a formal relationship between parallel trends and the graphical notion of confounding captured by sufficient adjustment sets.  Specifically, if $M$ represents a common (possibly unmeasured) sufficient set for $Y_1$ and $Y_0$, then the common dependence pattern specified in Lemma \ref{lemma:ghanemE3} applies to $M$ regardless of whether variables in $M$ are parents of $A$:
\begin{lemma}[Common dependence on a sufficient set]\label{lemma:common_dep} Let $M$ denote a common sufficient set for $Y_0$ and $Y_1$, such that  $\E(Y_t^0|A,M)=\E(Y_t^0|M)$ $a.s.$, $t=0,1$. Then, under a mild regularity condition (described in the appendix) parallel trends implies
\begin{equation}\label{eq:ahc}
    \E(Y_1^0 -Y_0^0|M)=\E(Y_1^0-Y_0^0) \quad a.s.
\end{equation}
\end{lemma}
A proof of Lemma \ref{lemma:common_dep} is provided in Appendix \ref{app:common_dep}. Similarly to (\ref{eq:ghanemE3}), an equivalent representation of (\ref{eq:ahc}) is $\E(\dot Y_1^0 |M)=\E(\dot Y_0^0|M)$; i.e., a common functional form for the dependence of $Y_t^0$ ($t=0,1$) on the common sufficient set $M$, at the mean.
Though Lemma \ref{lemma:common_dep} only applies to a common sufficient set, this turns out not to be restrictive. Specifically, we also prove in Appendix \ref{lemma:common_dep} that there is always a common sufficient set for $Y_0$ and $Y_1$ in any graph without an arrow from $Y_0$ to $A$, and in the next subsection we show why an arrow from $Y_0$ to $A$ is usually at odds with parallel trends.

The proof of Lemma \ref{lemma:common_dep} builds directly on \cite{ghanem2022selection}'s proof of Lemma \ref{lemma:ghanemE3}. Intuitively, since all confounding is delivered through backdoor paths \citep{greenland1999causal}, we may deduce from equations (\ref{eq:ec1}-\ref{eq:ec2}) that parallel trends implies all backdoor paths between $A$ and $Y_0$ deliver equal additive association as the set of backdoor paths between $A$ and $Y_1$. Lemma \ref{lemma:common_dep} formalizes this idea by stating that any common adjustment set must share an equal additive association pattern. 

\begin{remark}[Sufficiency of additive homogeneous confounding]
We will refer to (\ref{eq:ahc}) as \textit{additive homogeneous confounding} as this is an important relationship that many of our results build on. In addition to being necessary, (\ref{eq:ahc}) is also sufficient for parallel trends to hold. To see this, note that:
\begin{align*}
    \E(Y_1^0 - Y_0^0|A)&=\E\{\E(Y_1^0 - Y_0^0|M, A)|A\}\\
    &=\E\{\E(Y_1^0 - Y_0^0|M)|A\}\\
    &=\E\{\E(Y_1^0 - Y_0^0)|A\} \\
    &=\E(Y_1^0 - Y_0^0)
\end{align*}
using the sufficiency of $M$ in the second equality and (\ref{eq:ahc}) in the third.    
\end{remark}

The remainder of our results build on Lemmas \ref{lemma:ghanemE3} and \ref{lemma:common_dep} to describe when a graph can be rejected under Assumption \ref{asn:f}.

\subsection{Condition 1: No arrow from $Y_0$ to $A$}\label{sec:y0a}
Under Assumption \ref{asn:f}, parallel trends is incompatible with a graph containing both (i) an arrow from $Y_0$ to $A$ and (ii) unmeasured confounding (i.e., open paths between $A$ and $Y_1^0$). Specifically, note that in the presence of an arrow from $Y_0$ to $A$, from Lemma \ref{lemma:ghanemE3} it follows that parallel trends implies a conditional mean exchangeability condition, 
\begin{equation}\label{eq:exch}
    \E(Y_1^0|A, Y_0) = \E(Y_1^0|Y_0)
\end{equation}
To see this, note that if $Y_0$ is a parent of $A$, parallel trends implies:
\begin{align}\label{eq:noy0a}
    \E\{\dot Y_1^0 |pa(A)\} &= \E\{\dot Y_0^0|pa(A)\} = \dot Y_0^0
\end{align}
where the first equality is from Lemma \ref{lemma:ghanemE3} and the second equality always holds when $pa(A)$ includes $Y_0$ (the latter is equal to $Y_0^0$ under no anticipation). In other words, the inclusion of $Y_0$ in $pa(A)$ implies that $\E\{ Y_1^0 |pa(A)\}-\E(Y_1^0)$ (and hence $\E\{Y_1^0|pa(A)\}$) does not depend on any parents of $A$ besides $Y_0$. Then, conditional mean exchangeability (\ref{eq:exch}) holds because:
\begin{alignat*}{3}
    \E(\dot Y_1^0|A, Y_0) &= \E[\E\{ \dot Y_1^0|pa(A), A\}|A, Y_0] \\
     &= \E[\E\{ \dot Y_1^0|pa(A)\}|A, Y_0]\\
     &= \E[\E\{ \dot Y_0^0|pa(A)\}|A, Y_0]\\
     &= \dot Y_0^0,
\end{alignat*} 
where the second equality holds by the causal Markov assumption, and the third and fourth equalities apply (\ref{eq:noy0a}). 
Then from (\ref{eq:noy0a}), we have that $\E( \dot Y_1^0|A, Y_0)=\E( \dot Y_1^0| Y_0)$; adding $\E(Y_1^0)$ to both sides proves the result.

However, if there are any open paths between $A$ and $Y_1^0$, then linear faithfulness implies (\ref{eq:exch}) must be false. Additionally, even without adopting linear faithfulness, parallel trends in the presence of an arrow from $Y_0$ to $A$ would imply the ATT was identifiable by adjustment for $Y_0$ alone, so that DID would be unnecessary. 

\subsection{Condition 2: No disjoint sufficient sets for $Y_0$ and $Y_1$}\label{sec:distinct}
Under linear faithfulness, parallel trends implies that any sufficient set for $Y_0$ must typically also be a sufficient set for $Y_1$ and vice versa. The caveat ``typically'' is that the statement is only true for sufficient sets that are each a subset of a common sufficient set for both $Y_1$ and $Y_0$, though this is not really restrictive for reasons previously stated. We state this result as a Lemma to make it more clear.
\begin{lemma}[All sufficient sets in common]\label{lemma:allcommon}
Let $M$ denote a common sufficient adjustment set for $Y_0$ and $Y_1$, such that $\E(Y_t^0|A,M)=\E(Y_t^0|M)$ $a.s.$, $t=0,1$. Suppose $M_0 \subset M$ is a sufficient set for some $Y_s$ $(s\in \{0,1\})$, such that $\E(Y_s^0|A, M_0)=\E(Y_s^0|M_0)$. Then, parallel trends implies $M_0$ is also a common sufficient adjustment set for $Y_0$ and $Y_1$ such that  $\E(Y_t^0|A,M_0)=\E(Y_t^0|M_0)$, $t=0,1$.   
\end{lemma}

The proof of Lemma \ref{lemma:allcommon} is in Appendix \ref{app:allcommon}. Graphically, Lemma \ref{lemma:allcommon} implies a violation of linear faithfulness occurs whenever the conditional independence $\E(Y_t^0|A,M_0)=\E(Y_t^0|M_0)$ is not implied by the graph for time point $t\neq s$. In other words, if the causal diagram implies that $M_0\subset M$ is a sufficient adjustment set for $Y_0$ ($Y_1$), but not $Y_1$ ($Y_0$), then parallel trends implies a violation of linear faithfulness. For example, in the SWIG in Figure \ref{fig:swig2}, from d-separation rules we have the common sufficient set $M=\{U_1, U_3, U_4\}$, with $M_0=\{U_1, U_3\}$ sufficient for $Y_0$ but not $Y_1$. In other words, the graph implies that if we omit $U_4$ from $M$, we still have a sufficient set for $Y_0$, but we no longer have a sufficient set for $Y_1$. However, from Lemma \ref{lemma:allcommon} we have that parallel trends implies that $M_0$ must be sufficient for $Y_1$ as well; i.e., $U_4$ is not required in $M$ for it to be sufficient for either $Y_0$ and $Y_1$. Thus, due to the existence of $U_4$ in  Figure \ref{fig:swig2}, parallel trends cannot hold under linear faithfulness. By parallel arguments, the existence of $U_3$ in Figure \ref{fig:swig2} also implies parallel trends cannot hold under linear faithfulness. 

An important implication of Lemma \ref{lemma:allcommon} is that, if no common set of variables represents a minimally sufficient set for both outcome times, then parallel trends implies a violation of linear faithfulness; i.e., at least one backdoor path depicted in the graph must transmit zero total association at the mean. This is evident in Figure \ref{fig:swig2} where $\{U_1, U_3\}$ is the only minimally sufficient set for $Y_0$, and $\{U_1, U_4\}$ is the only minimally sufficient set for $Y_1$.

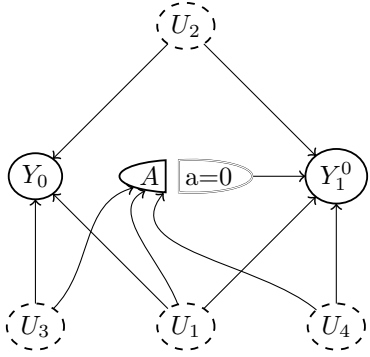
\begin{figure}
    [t!]
    \centering
    \begin{tikzpicture}[node distance=2cm]
\tikzset{line width=0.7pt, outer sep=0pt,
         ell/.style={draw,fill=white, inner sep=2pt,
          line width=0.7pt},
         swig vsplit={gap=5pt,
         inner line width right=0.5pt}}
        \node (y0) [ell, shape=ellipse] {$Y_0$};
        \node (a) [ell, shape=swig vsplit, right of=y0]  {\nodepart{left}{$A$}
        \nodepart{right}{a=0}};
        \node (y1) [ell, shape=ellipse, right of=a] {$Y_1^0$};
        \node (u1) [ell, dashed, shape=ellipse, below of=a] {$U_1$};
        \node (u2) [ell, dashed, shape=ellipse, above of=a] {$U_2$} ;
        \node (u3) [ell, dashed, shape=ellipse, below of=y0] {$U_3$};
        \node (u4) [ell, dashed, shape=ellipse, below of=y1] {$U_4$};
        \draw[->] (a) to (y1);
        \draw[->] (u1) to (y0);
        \draw[->] (u1) to[out=110, in=-155] (a);
        \draw[->] (u3) to[in=-165] (a);  
        \draw[->] (u4) to[out=140, in=-130] (a); 
        \draw[->] (u1) to (y1);
        \draw[->] (u2) to (y0);
        \draw[->] (u2) to (y1);
        \draw[->] (u3) to (y0);
        \draw[->] (u4) to (y1);
        \end{tikzpicture}
    \caption{A directed SWIG in which $Y_0$ and $Y_1$ contain the disjoint sufficient adjustment sets $\{U_1, U_3\}$ and $\{U_1, U_4\}$, respectively, and which are a each subset of the common adjustment set $\{U_1, U_3, U_4\}$.}
    \label{fig:swig2}
\end{figure}
\subsection{Condition 3: No arrow from $Y_0$ to $Y_1^0$}\label{sec:y0y1}
In this subsection, we argue that it may be undesirable to assume parallel trends in the presence of a $Y_0\rightarrow Y_1^0$ arrow, even though it would not result in a logical contradiction. At a high level, our argument is as follows. Suppose the structural model that generated our data is: (a) partially linear (i.e. one in which $Y_0$ enters additively separably into the structural equation for $Y_1^0$); (b) satisfies parallel trends; and (c) is described by a SWIG with a $Y_0\rightarrow Y_1^0$ arrow. Then it follows that parallel trends would no longer hold in a slightly modified structural model in which the $Y_0\rightarrow Y_1$ arrow is removed (by removing the additive $Y_0$ term from the structural equation for $Y_1^0$). The converse is also true. If parallel trends holds under a structural model corresponding to a SWIG with no $Y_0\rightarrow Y_1^0$ arrow, then it no longer holds after inserting the arrow by adding a separable additive $Y_0$ term to the structural equation for $Y_1^0$. 

In other words, if one believes parallel trends holds (under a partially linear model) in the presence of a $Y_0\rightarrow Y_1^0$ arrow, one must believe  that parallel trends holds \textit{because of} the arrow. The arrow must convey the precise magnitude of association required to offset the violation of parallel trends that would be present under a model without the arrow. We cannot imagine a reasonable mechanism by which the effect of $Y_0$ on $Y_1^0$ should be more than loosely connected to the magnitude of the parallel trends violation in an otherwise similar model lacking the $Y_0\rightarrow Y_1^0$ arrow, let alone exactly offsetting it. Therefore, (under the 
partial linearity assumption) parallel trends in the presence of a $Y_0\rightarrow Y_1^0$ amounts to a remarkable coincidence. We conjecture that similar coincidences would be required in nonparametric models. In the absence of the arrow, on the other hand, (and under linear faithfulness and the other graphical conditions) parallel trends reduces to additive homogeneous confounding, which might be justified more naturally on grounds of stability of (albeit scale specific) confounding effects over time.%but contains neither (i) a $Y_0\rightarrow A$ arrow, nor (ii) disjoint sufficient sets for $Y_0$ and $Y_1$. Thus, our SWIG does not contradict parallel trends under the conditions stated so far, and parallel trends will hold if and only if additive homogeneous confounding (\ref{eq:ahc}) holds. Then, under a partially linear model for $Y_1^0$ given $Y_0$, deleting the arrow $Y_0\rightarrow Y_1^0$ causes the SWIG to contradict parallel trends. On the other hand, if our initial SWIG does not contain the arrow $Y_0\rightarrow Y_1^0$ and otherwise does not contradict parallel trends, then adding a $Y_0\rightarrow Y_1^0$ arrow (parameterized as a model for $Y_1^0$ that is partially linear in $Y_0$) breaks parallel trends. 

%This suggests that we may not be agnostic to the $Y_0\rightarrow Y_1^0$ arrow and still assume parallel trends. We examine what a scientific argument for (\ref{eq:ahc}) would require in the case (i) with and (ii) without this arrow. In case (i), the $Y_0\rightarrow Y_1^0$ arrow contributes to confounding only by transmitting the association between $M$ and $Y_0$. Therefore, the $Y_0\rightarrow Y_1^0$ must serve to balance association by the remaining paths between $M$ and $Y_1^0$ in order for (\ref{eq:ahc}) to hold. Such balancing is not required in case (ii), making the justification appear easier. The caveat is that these arguments depend on the model being semiparametric with outcome models that are partially linear in $Y_0$ and/or $M$. 

Now, we will formalize the argument that one cannot be agnostic to the $Y_0\rightarrow Y_1^0$ arrow. Consider two models $G_0$ and $G_1$. Both share the same structure as the NPSEM in Figure \ref{fig:npsem} except that $A=f_A^0(U_1, U_3, U_4, \varepsilon_A)$ in both (so that $Y_0$ does not affect $A$), and $f_{Y_1}$ varies between the two as:
\begin{align*}
    G_0:Y_1^0&=f^0_{Y_1}(U_1, U_2, U_4,\varepsilon_{Y_1})\\
    G_1:Y_1^0&=f^0_{Y_1}(U_1, U_2, U_4, \varepsilon_{Y_1}) + \alpha Y_0 ;
\end{align*}
that is, $Y_0$ does not affect $Y_1$ in $G_0$; and in $G_1$, $Y_0$ affects $Y_1$ linearly via $\alpha$. Parallel trends holding in both models represents a violation of linear faithfulness, which we will now demonstrate. 

\begin{proof}
    Begin by assuming parallel trends holds in both $G_1$ and $G_0$. In both models, it can be shown that there exists the common sufficient set $M=\{U_1, U_3, U_4\}$ (see Appendix). Then, from Lemma \ref{lemma:common_dep}, parallel trends holding in both models implies:
\begin{equation}\label{eq:balance}
    \E(\dot Y_0^0|M)=\E_{G_0}(\dot Y_1^0|M)=\E_{G_1}(\dot Y_1^0|M) \quad a.s.
\end{equation}
where $\E_G$ denotes that the expectation is taken over model $G$ (note that the model for $Y_0^0$ is equivalent across $G_1$ and $G_0$ so that $\E_{G_1}(\dot Y_0^0|M)=\E_{G_0}(\dot Y_0^0|M)=\E(\dot Y_0^0|M)$ by definition). Then, note that $\E_{G_1}(\dot Y_1^0|M)=\E\{\dot f^0_{Y_1}(U_1, U_2, U_4, \varepsilon_{Y_1}) + \alpha \dot Y_0 |M\}=\E_{G_0}(\dot Y_1^0|M) + \alpha \E( \dot Y_0^0 |M)$ $a.s.$, so that parallel trends in both models implies:
\begin{equation}\label{eq:wierd}
    \E(\dot Y_0^0|M)-\E_{G_0}(\dot Y_1^0|M)=\alpha \E( \dot Y_0^0 |M) \quad a.s.
\end{equation}
which, from (\ref{eq:balance}), implies $\alpha \E( \dot Y_0^0 |M)=0$, a violation of linear faithfulness. 
\end{proof}

If parallel trends cannot hold in both $G_0$ and $G_1$, which model should one prefer? We argue that one should prefer $G_0$ for reasons we will now explain. For parallel trends to hold in $G_1$ requires special balancing of associations, which is evident from equation (\ref{eq:wierd}) as well. In a slight abuse of notation, conceptualize $\E_G(\dot Y_t^0|M)$ as the time-$t$ confounding in model $G$. Suppose parallel trends holds in $G_1$ but not $G_0$. The time-1 confounding in $G_1$ can be decomposed as:
\begin{align*}
    \E_{G_1}(\dot Y_1^0|M) &= \underbrace{\E_{G_1}(\dot Y_1^0|M)-\E_{G_0}(\dot Y_1^0|M)}_{\Delta_1}+\underbrace{\E_{G_0}(\dot Y_1^0|M)}_{\Delta_0}.
\end{align*}
Since the only difference between $G_1$ and $G_0$ is the $Y_0\rightarrow Y_1^0$ arrow, $\Delta_1$ captures the time 1 confounding passing through $Y_0$ in $G_1$, and $\Delta_0$ captures the time 1 confounding \textit{not} passing through $Y_0$ in $G_1$. Note that parallel trends not holding in $G_0$ implies that $\Delta_0\neq \E(\dot Y_0^0|M)$; i.e., the time-1 confounding not passing through $Y_0$ differs from the time-0 confounding by some amount. Equation (\ref{eq:wierd}) states that $\alpha$ must precisely multiply with the time 0 confounding to offset this latter difference. To illustrate why this may be seen as a remarkable coincidence, Figure \ref{fig:swig5} provides SWIG for $G_1$ where for visual simplicity $M$ is collapsed to single node. For parallel trends to be violated in $G_0$, associations transmitted by the edges $\pi_0$ and $\pi_1$ differ by some amount. For parallel trends to hold in $G_1$, $\alpha$ must combine with $\pi_0$ to exactly cancel out this difference, suggesting that $\alpha$ would need to contain information about how much the associations transmitted along the edges $\pi_0$ and $\pi_1$ differ. On the other hand, such balancing is not needed for parallel trends to hold in $G_0$, as only associations transmitted along the edges $\pi_0$ and $\pi_1$ need to be equal.

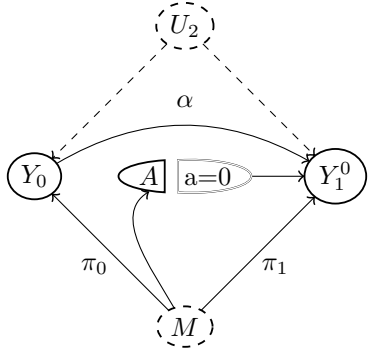
\begin{figure}
[!t]
    \centering
    \begin{tikzpicture}[node distance=2cm]
\tikzset{line width=0.7pt, outer sep=0pt,
         ell/.style={draw,fill=white, inner sep=2pt,
          line width=0.7pt},
         swig vsplit={gap=5pt,
         inner line width right=0.5pt}}
       \node (y0) [ell, shape=ellipse] {$Y_0$};
        \node (a) [ell, shape=swig vsplit, right of=y0]  {\nodepart{left}{$A$}
        \nodepart{right}{a=0}};
        \node (y1) [ell, shape=ellipse, right of=a] {$Y_1^0$};
        \node (u1) [ell, dashed, shape=ellipse, below of=a] {$M$};
        \node (u2) [ell, dashed, shape=ellipse, above of=a] {$U_2$};
        \node (pi1) [below of=y1, yshift=8mm, xshift=-8mm] {$\pi_1$};
        \node (pi0) [below of=y0, yshift=8mm, xshift=8mm] {$\pi_0$};
        \node (alpha) [above of=a, yshift=-1cm] {$\alpha$};
        \draw[->] (a) to (y1);
        \draw[->] (u1) to (y0);
        \draw[->] (u1) to[out=120, in=-150] (a);
        \draw[->] (u1) to (y1);
        \draw[->] (y0) to[bend left=30] (y1);
        \draw[->, dashed] (u2) to (y0);
        \draw[->, dashed] (u2) to (y1);
        \end{tikzpicture}
    \caption{Simplified SWIG representing $G_1$}
    \label{fig:swig5}
\end{figure}

In nonparametric models, the way in which the time-0 confounding must combine with the $Y_0\rightarrow Y_1^0$ arrow to guarantee parallel trends may be more complex. As a result, assuming that parallel trends would hold in models with and without this arrow may not present as a violation of linear faithfulness, but we conjecture that a similar strange cancellation would generally be required. For example, consider an additively separable rather than partially linear model; i.e., we might consider $G_1':Y_1=f^0_{Y_1}(U_1, U_2, U_4, A,\varepsilon_{Y_1}) + h(Y_0)$. Parallel trends holding in both $G_0$ and $G_1'$ implies $\E[h(Y_0)-\E\{h(Y_0)\}|M]=0$ $a.s.$; i.e., $h(Y_0)$ must uncorrelated with $M$ even though d-separation rules do not imply that $Y_0$ and $M$ are uncorrelated. Thus, $h(\cdot)$ must somehow omit all confounding variation contained in $Y_0$. While not strictly a violation of linear faithfulness in this more general case, the requirement that an edge only transmit non-confounding variation seems highly restrictive. 

%Generally, we find it troubling that, under these models, adding a $Y_0$ to $Y_1$ arrow always breaks parallel trends (or implies implausible restrictions) if parallel trends is present without the arrow. Put another way, under these models, if parallel trends holds in the presence of a $Y_0$ to $Y_1$ arrow, it is because the arrow is conveying the precise amount of association required to compensate for a disparity in the magnitude of confounding conveyed directly through the unobserved confounders at the two time points. This would seem a remarkable coincidence. We conjecture that similar coincidences would be required in nonparametric models.

\subsection{Summary of conditions}\label{sec:summary}
In the previous three subsections, we have argued that under assumptions which are likely reasonable in practice, parallel trends implies the following conditions:
\begin{condition}
$Y_0$ does not directly affect $A$ so long as unmeasured confounding exists between $A$ and $Y_1$
\end{condition}
\begin{condition}
There do not exist disjoint sufficient sets for $Y_0$ and $Y_1$, which are subsets of a common sufficient set
\end{condition} 
\begin{condition}
$Y_0$ does not directly affect $Y_1$
\end{condition}
Specifically, Conditions 1 and 2 are implied by parallel trends under linear faithfulness alone. Though Condition 3 is not in general implied under linear faithfulness in a nonparametric model, in some reasonable semiparametric models, parallel trends cannot hold both with and without a $Y_0\rightarrow Y_1$ arrow (under linear faithfulness). Moreover, the existence of a $Y_0\rightarrow Y_1$ edge in more flexible models implies strange conditions that are similar to, but not formally included under, violations of linear faithfulness. Thus, we might reasonably take Condition 3 to apply more broadly. 

The implications of Condition 2 are somewhat more subtle than, and are interrelated with, the other two conditions. First, if there is an edge $Y_0 \rightarrow A$ then there is no sufficient set for $Y_0$ and so Condition 2 is not applicable. Second, if there is an edge $Y_0\rightarrow Y_1$, then $Y_0$ and $Y_1$ might share minimally sufficient sets that were not otherwise shared. Specifically, with the $Y_0\rightarrow Y_1$ edge, any common causes of $A$ and $Y_0$ are also common causes of $A$ and $Y_1$; these variables will generally appear in a minimally sufficient set for both outcome times, when they may not have for $Y_1$ otherwise. Thus, if one does not accept Condition 3 as being universally applicable, then in a given application we may find violations of Condition 2 which can appear to be solved by assuming that the edge $Y_0\rightarrow Y_1$ exists. The idea that the existence of certain arrows could make parallel trends possible when it would be rejected otherwise appears to us implausible. This point may be seen as providing further evidence that Condition 3 is broadly applicable. 

If one adopts all three conditions, and if one wishes to allow for unmeasured confounding and for $Y_0$ and $Y_1$ to be correlated, the SWIG in Figure \ref{fig:swig4} can be shown to be the maximal graph compatible with parallel trends. To see this, in Figure \ref{fig:swig3}, we have maintained all edges from Figure \ref{fig:npsem} except for the $Y_0\rightarrow A$ and $Y_0\rightarrow Y_1^0$ edges (that is, we adopt Conditions 1 and 2). Then, by d-separation rules, we have that in all directed SWIGs compatible with the partially directed SWIG Figure \ref{fig:swig3},  $Y_1$ has the minimally sufficient sets $M_0=\{U_1, U_3, U_4\}$ and $M_1=\{U_1, U_2, U_4\}$ , whereas $Y_0$ has only one, $M_0=\{U_1, U_3, U_4\}$; we have a violation of Condition 2. (We provide R code in the Appendix to demonstrate this using the \verb|dagitty| package in R; showing this analytically is complex and outside of our scope.) Set $M_1$ occurs because of paths going through $U_2$. Therefore, if we do not wish to remove $U_2$ (i.e., we wish to allow $Y_0$ and $Y_1$ to be correlated), we must delete one or more edges connecting $U_2$ to the other $U$s, and/or the $Y_0\rightarrow U_4$ edge. Any such deletions that lead to the removal of the minimally sufficient set $M_1$ will result in $Y_0$ and $Y_1$ each having a single minimally sufficient set, $\{U_1, U_3\}$ and $\{U_1, U_4\}$, respectively (as in Figure \ref{fig:swig3}). The only way to ensure that $Y_0$ and $Y_1$ share all minimally sufficient sets is to assume that $U_3$ and $U_4$ do not exist. Thus, Conditions 1-3 imply the confounding mechanism in DID must be representable by Figure \ref{fig:swig4}. 
\begin{figure}
    [t!]
    \centering
    \begin{tikzpicture}[node distance=2cm]
\tikzset{line width=0.7pt, outer sep=0pt,
         ell/.style={draw,fill=white, inner sep=2pt,
          line width=0.7pt},
         swig vsplit={gap=5pt,
         inner line width right=0.5pt}}
        \node (y0) [ell, shape=ellipse] {$Y_0$};
        \node (a) [ell, shape=swig vsplit, right of=y0]  {\nodepart{left}{$A$}
        \nodepart{right}{a=0}};
        \node (y1) [ell, shape=ellipse, right of=a] {$Y_1^0$};
        \node (u1) [ell, dashed, shape=ellipse, below of=a] {$U_1$};
        \node (u2) [ell, dashed, shape=ellipse, above of=a] {$U_2$} ;
        \node (u3) [ell, dashed, shape=ellipse, below of=y0] {$U_3$};
        \node (u4) [ell, dashed, shape=ellipse, below of=y1] {$U_4$};
        \draw[->] (a) to (y1);
        \draw[->] (u1) to[out=110, in=-155] (a);
        \draw[->] (u3) to[in=-165] (a);  
        \draw[->] (u4) to[out=140, in=-130] (a);
        \draw[->] (u1) to[bend left=10] (y0);
        \draw[->] (u1) to[bend right=10] (y1);
        \draw[->] (u2) to[bend right =10] (y0);
        \draw[->] (u2) to (y1);
        \draw[->] (u3) to (y0);
        \draw[->] (u4) to (y1);
        \draw[->] (u1) to (u4);
        \draw[->] (y0) to[bend right=10] (u4);
        \draw[dashed, -] (u1) to [bend left=45] (u2);
        \draw[dashed,-] (u1) to (u3);
        \draw[dashed,-] (u1) to (u4);
        \draw[dashed,-] (u2) to[bend right=75] (u3);
        \draw[->] (u2) to[bend left=75] (u4);
        \draw[->] (u3) to [bend right=23] (u4);
        \end{tikzpicture}
    \caption{A partially directed SWIG in which parallel trends would be rejected under linear faithfulness, as $U_4$ appears in a minimally sufficient set for $Y_1$ but not $Y_0$.}
    \label{fig:swig3}
\end{figure}
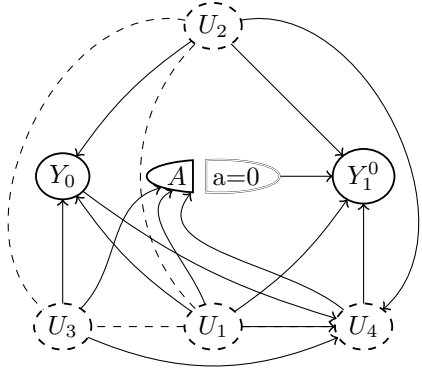
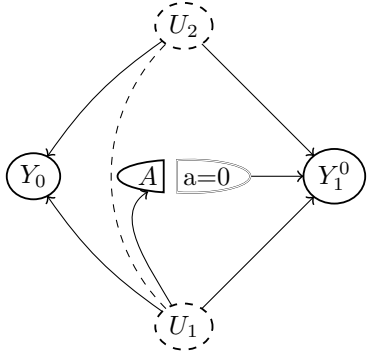
\begin{figure}[t]
     \centering
    \begin{tikzpicture}[node distance=2cm]
\tikzset{line width=0.7pt, outer sep=0pt,
         ell/.style={draw,fill=white, inner sep=2pt,
          line width=0.7pt},
         swig vsplit={gap=5pt,
         inner line width right=0.5pt}}
       \node (y0) [ell, shape=ellipse] {$Y_0$};
        \node (a) [ell, shape=swig vsplit, right of=y0]  {\nodepart{left}{$A$}
        \nodepart{right}{a=0}};
        \node (y1) [ell, shape=ellipse, right of=a] {$Y_1^0$};
        \node (u1) [ell, dashed, shape=ellipse, below of=a] {$U_1$};
        \node (u2) [ell, dashed, shape=ellipse, above of=a] {$U_2$} ;
        \draw[->] (a) to (y1);
        \draw[->] (u1) to[bend left=10] (y0);
        \draw[->] (u1) to[out=120, in=-150] (a);
        \draw[->] (u1) to (y1);
        \draw[->] (u2) to[bend right =10] (y0);
        \draw[->] (u2) to (y1);
        \draw[dashed, -] (u1) to [bend left=45] (u2);
        \end{tikzpicture}
    \caption{Maximal SWIG consistent with conditions 1, 2, and 3; when allowing for (i) correlation between $Y_0$ and $Y_1$, and (ii) unmeasured confounding of the relation between $A$ and $Y_1$.}
    \label{fig:swig4}
\end{figure}

Thus, in most applications, the maximal graph compatible with parallel trends would be the graph shown in Figure \ref{fig:swig4}, in which a common set of variables $U_1$ confounds both $Y_0$ and $Y_1$, and there are no effects of $Y_0$ on either the treatment or the post-treatment outcome. Importantly, $Y_0$ and $Y_1$ are allowed to be correlated through common unmeasured causes $U_2$. Notably, model (\ref{eq:twfe}) could be represented by the graph in Figure \ref{fig:swig4} if we substitute $\alpha_i=U_{i1}$ and $\varepsilon_{it}=U_{i2}+\varepsilon_{it}^*$, where $\varepsilon_{it}^*$ is an independent error. However, Figure \ref{fig:swig4} does not imply the parametric restrictions of model (\ref{eq:twfe}). 

Note that we have mainly discussed necessary implications of (not sufficient conditions for) parallel trends. Parallel trends is indeed not implied by the SWIG in Figure \ref{fig:swig4} or by any NPSEM due to the scale-dependency of the assumption. However, Remark 1 indicates that additive homogeneous confounding (\ref{eq:ahc}) is both necessary and sufficient.  Since this is a scale-dependent restriction, it will often be difficult to justify in practice. However, if, as we have argued is usually necessary, our SWIG corresponds to Figure \ref{fig:swig4}, then (\ref{eq:ahc}) can be interpreted loosely as stating that the arrows $U_1\rightarrow Y_0$ and $U_1\rightarrow Y_1$ deliver equal additive association, which may be reasonable to argue in some applications. Note that similar assumptions about constancy of associations across time and place are often made in other settings; for example, in quantitative bias analyses of unmeasured confounding based on an external validation study \citep{lash2014good}, or when transporting study results to new settings \citep{dahabreh2019extending}.

\section{Plausibility of Parallel Trends for Medicaid Expansion}
Here, we apply our approach to attempt to falsify parallel trends for the motivating example. We begin by considering whether (i) a $Y_0 \rightarrow A$ exists. Though states may have been motivated by uninsurance rates, we expect that expansion decisions were almost exclusively driven by their political lean in the polarized political climate surrounding the ACA. Next, regarding (ii) the potential existence of non-overlapping minimally sufficient sets, we consider whether there may be disjoint common causes of $A$ and $Y_0$ on the one hand, and $A$ and $Y_1$ on the other. Again, we expect the main driving force in determining $A$ to be the state political lean, which could affect insurance rates in times 0 and 1 via other policies (such as labor policy), thus being a component of $U_1$ in Figure \ref{fig:npsem}. Another component of $U_1$ might be the geographic location of the state, which, through historical processes, could both cause political lean and influence local economic conditions, which could in turn impact uninsurance rates.  We are unaware of variables such as $U_3$ and $U_4$ that could have an effect only in one year. Such variables could represent temporary economic shocks. For example, the importance of farming in a state's economy ($U_3$) might influence the expansion decision through broad impacts on the state's economy and culture. Then, supposing a temporary shock (such as a drought) occurred in 2013, $U_3$ might affect insurance rates only in 2013 via a large number of farmers falling below 138\% of the federal poverty level, but not in 2014 when the economy may have recovered. We are unaware of such a differential temporal shock in this example and only present it as a possibility. Finally, regarding $Y_0 \rightarrow Y_1$, though we expect insurance rates to be correlated year-on-year, an intervention that could shift insurance rates in one year but not the next (for example, an emergency Medicaid expansion as occurred during the COVID-19 pandemic) would likely have limited impact on insurance rates in the following year. Thus we could reasonably argue that the SWIG in Figure \ref{fig:swig4} describes the DGM, which would not lead us to reject parallel trends under linear faithfulness alone. 

However, we would still need to justify additive homogeneous confounding. Would pre-expansion political climate ($U_1$)  have constant additive association with insurance rates in the absence of intervention? Due to the scale-dependence of this assumption and the fact that insurance rates change year-to-year, we are unaware of a way to justify it a priori. 

\section{Discussion}
Since parallel trends is a scale-dependent assumption, it does not imply any restrictions on a causal DAG without additional assumptions. However, under a linear faithfulness assumption stating that d-connected variables must be correlated, we showed that certain features of a causal DAG can be in contradiction to parallel trends. In particular, a causal DAG is in conflict with parallel trends under linear faithfulness if either (i) pre-treatment outcomes affect treatment assignment, or (ii) the unmeasured confounding variables do not form a common minimally sufficient set to identify the effect of treatment on pre- and post-treatment outcomes. Additionally, (iii) an arrow from pre- to post-treatment outcomes raises concerns for parallel trends as it can be shown to violate linear faithfulness in some relatively simple models. Thus, one can decide whether a given application should be under consideration for DID based on a causal diagram. If a causal diagram based on substantive knowledge includes any of these three features, parallel trends is unlikely to hold in this application and other methods might be better suited. Supposing one does not reject parallel trends based on a causal diagram, one still must justify additive homogeneous confounding. This is a scale dependent assumption (i.e., homogeneity on the additive scale usually implies heterogeneity on the multiplicative scale, etc.) and thus is usually hard to justify. In some settings, it may be possible to find support for this assumption through external empirical studies on the relation between the posited unmeasured confounders and the outcomes over time. 

The standard approach to justifying parallel trends in applied literature is to compare empirical trends in the treated and untreated over several periods prior to the treatment \citep{cunningham2021causal, angrist2009mostly}. Our approach to justifying parallel trends can complement this purely empirical approach by incorporating a priori subject matter knowledge in a DAG. The relationship between pre-trends and structural relationships representable in a DAG will be a subject of future work.

Broadly, our results generalize and connect principles that have been noted in the literature, mainly in the context of linear structural models \citep{ kim2021gain,zhang2021exploiting,weber2015assumption,zeldow2021confounding,wooldridge2010econometric}. In particular, Condition 1 has been noted repeatedly as important to ensuring parallel trends \citep{weber2015assumption, kim2021gain}, and has been considered as part of the ``strict exogeneity'' assumption adopted in linear fixed effects models \citep{wooldridge2010econometric}. While Condition 2 has not been fully stated to our knowledge, \cite{kim2021gain} noted in a linear setting the potential for violations of parallel trends based on common causes of pre-treatment outcomes and treatment that are not causes of post-treatment outcomes. One may also consider violations of Condition 2 to be a type of time-varying unmeasured confounding, which is usually thought to be incompatible with parallel trends \citep{zeldow2021confounding}. Regarding Condition 3, \cite{kim2021gain} noted (in a linear setting) that for parallel trends to hold in the presence of an arrow from pre- to post-treatment outcomes, it would have to be the case that the product of multiple path coefficients equal a single path coefficient, which is similar to our finding that, under linear faithfulness and certain semiparametric restrictions on the model, parallel trends cannot hold both with and without such an arrow. Similarly, \cite{weber2015assumption} depicted the increment $Y^{\theta}=Y_1-Y_0$ on a graph and showed that in their example, and arrow  $Y_0\rightarrow Y^\theta$ violated parallel trends. This finding also supports the notion that an arrow $Y_0\rightarrow Y_1$ should be expected to violate parallel trends, since most data-generating mechanisms consistent with this arrow would also have an arrow $Y_0\rightarrow Y^\theta$.

Our analysis has important limitations. One is that many of our conditions depend on whether there exists a causal effect of variables for which a well-defined intervention may not exist . In particular, it may be difficult to reason about causal effects of prior outcomes $Y_0$; when these variables are complex processes like unemployment or uninsurance rates, many possible interventions exist. Additionally, we did not consider the role of measured covariates or of DID designs involving variation in treatment timing; these will be the subject of future work.

\bibliography{main}

\begin{thebibliography}{25}
\providecommand{\natexlab}[1]{#1}
\providecommand{\url}[1]{\texttt{#1}}
\expandafter\ifx\csname urlstyle\endcsname\relax
  \providecommand{\doi}[1]{doi: #1}\else
  \providecommand{\doi}{doi: \begingroup \urlstyle{rm}\Url}\fi

\bibitem[Abadie(2005)]{abadie2005semiparametric}
A.~Abadie.
\newblock Semiparametric difference-in-differences estimators.
\newblock \emph{The review of economic studies}, 72\penalty0 (1):\penalty0
  1--19, 2005.

\bibitem[Andersen(2013)]{andersen2013expect}
H.~Andersen.
\newblock When to expect violations of causal faithfulness and why it matters.
\newblock \emph{Philosophy of Science}, 80\penalty0 (5):\penalty0 672--683,
  2013.

\bibitem[Angrist and Pischke(2009)]{angrist2009mostly}
J.~D. Angrist and J.-S. Pischke.
\newblock \emph{Mostly harmless econometrics: An empiricist's companion}.
\newblock Princeton university press, 2009.

\bibitem[Ashenfelter(1978)]{ashenfelter1978estimating}
O.~Ashenfelter.
\newblock Estimating the effect of training programs on earnings.
\newblock \emph{The Review of Economics and Statistics}, pages 47--57, 1978.

\bibitem[Callaway(2023)]{callaway2023difference}
B.~Callaway.
\newblock Difference-in-differences for policy evaluation.
\newblock \emph{Handbook of Labor, Human Resources and Population Economics},
  pages 1--61, 2023.

\bibitem[Cunningham(2021)]{cunningham2021causal}
S.~Cunningham.
\newblock \emph{Causal inference: The mixtape}.
\newblock Yale university press, 2021.

\bibitem[Dahabreh and Hern{\'a}n(2019)]{dahabreh2019extending}
I.~J. Dahabreh and M.~A. Hern{\'a}n.
\newblock Extending inferences from a randomized trial to a target population.
\newblock \emph{European journal of epidemiology}, 34:\penalty0 719--722, 2019.

\bibitem[Ghanem et~al.(2022)Ghanem, Sant'Anna, and
  W{\"u}thrich]{ghanem2022selection}
D.~Ghanem, P.~H. Sant'Anna, and K.~W{\"u}thrich.
\newblock Selection and parallel trends.
\newblock \emph{arXiv preprint arXiv:2203.09001}, 2022.

\bibitem[Greenland et~al.(1999)Greenland, Pearl, and
  Robins]{greenland1999causal}
S.~Greenland, J.~Pearl, and J.~M. Robins.
\newblock Causal diagrams for epidemiologic research.
\newblock \emph{Epidemiology}, 10\penalty0 (1):\penalty0 37--48, 1999.

\bibitem[Kim and Steiner(2021)]{kim2021gain}
Y.~Kim and P.~M. Steiner.
\newblock Gain scores revisited: A graphical models perspective.
\newblock \emph{Sociological Methods \& Research}, 50\penalty0 (3):\penalty0
  1353--1375, 2021.

\bibitem[Lash et~al.(2014)Lash, Fox, MacLehose, Maldonado, McCandless, and
  Greenland]{lash2014good}
T.~L. Lash, M.~P. Fox, R.~F. MacLehose, G.~Maldonado, L.~C. McCandless, and
  S.~Greenland.
\newblock Good practices for quantitative bias analysis.
\newblock \emph{International journal of epidemiology}, 43\penalty0
  (6):\penalty0 1969--1985, 2014.

\bibitem[Lechner et~al.(2011)]{lechner2011estimation}
M.~Lechner et~al.
\newblock The estimation of causal effects by difference-in-difference methods.
\newblock \emph{Foundations and Trends{\textregistered} in Econometrics},
  4\penalty0 (3):\penalty0 165--224, 2011.

\bibitem[Pearl(2009)]{pearl2009causality}
J.~Pearl.
\newblock \emph{Causality}.
\newblock Cambridge university press, 2009.

\bibitem[Richardson and Robins(2013)]{richardson2013single}
T.~S. Richardson and J.~M. Robins.
\newblock Single world intervention graphs (swigs): A unification of the
  counterfactual and graphical approaches to causality.
\newblock \emph{Center for the Statistics and the Social Sciences, University
  of Washington Series. Working Paper}, 128\penalty0 (30):\penalty0 2013, 2013.

\bibitem[Rosenbaum and Rubin(1983)]{rosenbaum1983central}
P.~R. Rosenbaum and D.~B. Rubin.
\newblock The central role of the propensity score in observational studies for
  causal effects.
\newblock \emph{Biometrika}, 70\penalty0 (1):\penalty0 41--55, 1983.

\bibitem[Shpitser et~al.(2012)Shpitser, VanderWeele, and
  Robins]{shpitser2012validity}
I.~Shpitser, T.~VanderWeele, and J.~M. Robins.
\newblock On the validity of covariate adjustment for estimating causal
  effects.
\newblock \emph{arXiv preprint arXiv:1203.3515}, 2012.

\bibitem[Shrier and Platt(2008)]{shrier2008reducing}
I.~Shrier and R.~W. Platt.
\newblock Reducing bias through directed acyclic graphs.
\newblock \emph{BMC medical research methodology}, 8:\penalty0 1--15, 2008.

\bibitem[Sofer et~al.(2016)Sofer, Richardson, Colicino, Schwartz, and
  Tchetgen]{sofer2016negative}
T.~Sofer, D.~B. Richardson, E.~Colicino, J.~Schwartz, and E.~J.~T. Tchetgen.
\newblock On negative outcome control of unobserved confounding as a
  generalization of difference-in-differences.
\newblock \emph{Statistical science: a review journal of the Institute of
  Mathematical Statistics}, 31\penalty0 (3):\penalty0 348, 2016.

\bibitem[Spirtes et~al.(2001)Spirtes, Glymour, and
  Scheines]{spirtes2001causation}
P.~Spirtes, C.~Glymour, and R.~Scheines.
\newblock \emph{Causation, prediction, and search}.
\newblock MIT press, 2001.

\bibitem[Steel(2006)]{steel2006homogeneity}
D.~Steel.
\newblock Homogeneity, selection, and the faithfulness condition.
\newblock \emph{Minds and Machines}, 16:\penalty0 303--317, 2006.

\bibitem[Weber et~al.(2015)Weber, van~der Laan, and
  Petersen]{weber2015assumption}
A.~M. Weber, M.~J. van~der Laan, and M.~L. Petersen.
\newblock Assumption trade-offs when choosing identification strategies for
  pre-post treatment effect estimation: an illustration of a community-based
  intervention in madagascar.
\newblock \emph{Journal of causal inference}, 3\penalty0 (1):\penalty0
  109--130, 2015.

\bibitem[Wooldridge(2010)]{wooldridge2010econometric}
J.~M. Wooldridge.
\newblock \emph{Econometric Analysis of Cross Section and Panel Data}.
\newblock MIT Press, 2010.

\bibitem[Wooldridge(2021)]{wooldridge2021two}
J.~M. Wooldridge.
\newblock Two-way fixed effects, the two-way mundlak regression, and
  difference-in-differences estimators.
\newblock \emph{Available at SSRN 3906345}, 2021.

\bibitem[Zeldow and Hatfield(2021)]{zeldow2021confounding}
B.~Zeldow and L.~A. Hatfield.
\newblock Confounding and regression adjustment in difference-in-differences
  studies.
\newblock \emph{Health services research}, 56\penalty0 (5):\penalty0 932--941,
  2021.

\bibitem[Zhang et~al.(2021)Zhang, Cinelli, Chen, and
  Pearl]{zhang2021exploiting}
C.~Zhang, C.~Cinelli, B.~Chen, and J.~Pearl.
\newblock Exploiting equality constraints in causal inference.
\newblock In \emph{International Conference on Artificial Intelligence and
  Statistics}, pages 1630--1638. PMLR, 2021.

\end{thebibliography}

\begin{appendices}

\section{Proof of Lemma \ref{lemma:common_dep}}\label{app:common_dep}
\begin{proof}
   The proof follows directly from Ghanem et al. (2024, Lemma F.3). Let $\pi (M)=\E[A|M]$ denote the true propensity score. From \cite{rosenbaum1983central}, a property of the propensity score is that $A \ind \{Y_1^0, Y_0^0\} | \pi(M)$. Since it is always possible to write $A=\pi(M)+\epsilon_A$, we additionally have  $\{\pi(M) + \epsilon_A \}\ind \{Y_1^0, Y_0^0\} | \pi(M)$, further implying $\epsilon_A \ind (Y_1^0, Y_0^0) | \pi(M)$ . Note, from Ghanem et al. ( 2024, Lemma F.3), parallel trends implies $\E\{A(\dot Y_1^0 - \dot Y_0^0\}=0$. Then we have
\begin{align*}
0 &= \E\{A(\dot Y_1^0 - \dot Y_0^0)\}\\
&= \E[\{\pi(M) + \epsilon_A\}(\dot Y_1^0 - \dot Y_0^0)]\\
&= \E\{\pi(M)(\dot Y_1^0 - \dot Y_0^0)\}\\
&= \E\{\pi(M)\E(\dot Y_1^0 - \dot Y_0^0|M)\}
\end{align*}
where the third equality follows by the property of the propensity score that $\epsilon_A \ind (Y_1^0, Y_0^0) | \pi(M)$. Assume the following mild regularity condition:
\begin{assumption}[Varying conditional trends]\label{asn:varying} At least one of the following is true:
\begin{enumerate}
    \item $\Pr\{\E(Y_1^0-Y_0^0|M)>0\}<1$, or
    \item $\Pr\{\E(Y_1^0-Y_0^0|M)<0\}<1$
\end{enumerate}
\end{assumption}
i.e., the potential outcome trends, conditional on $M$, are either not all positive or not all negative. We first prove the case where Assumption \ref{asn:varying}(1) holds. If parallel trends holds for all $\pi(M)$, it must hold for $\pi(M)=\1\{\E(\dot Y_1^0 - Y_0^0|M) \geq 0\}$. For the latter choice of $\pi(M)$ we have:
\begin{align*}
    0 &= \E\{\pi(M)\E(\dot Y_1^0 - \dot Y_0^0|M)\} \\
    &= \E[\1\{\E(\dot Y_1^0 - Y_0^0|M) \geq 0\}\E(\dot Y_1^0 - \dot Y_0^0|M)]\\
    &=\E(\dot Y_1^0 - \dot Y_0^0|M)
\end{align*}
where the last equality follows from Lemma F.1 in \cite{ghanem2022selection}.

Next we prove the case were Assumption \ref{asn:varying}(2) holds. If parallel trends holds for all $\pi(M)$, it must hold for $\pi(M)=\1\{\E(\dot Y_1^0 - Y_0^0|M) \leq 0\}$. For the latter choice of $\pi(M)$ we have:
\begin{align*}
    0 &= \E\{\pi(M)\E(\dot Y_1^0 - \dot Y_0^0|M)\} \\
    &= \E[\1\{\E(\dot Y_1^0 - Y_0^0|M) \leq 0\}\E(\dot Y_1^0 - \dot Y_0^0|M)]\\
    &=\E(\dot Y_1^0 - \dot Y_0^0|M)
\end{align*}
where the last equality follows from Lemma F.1 in \cite{ghanem2022selection}.
\end{proof}

\section{Proof of existence of a common sufficient set}
Several of our results refer to the existence of a common sufficient set for the effect of $A$ on $\{Y_1^0, Y_0^0\}$. Here, we prove that this is not a restriction, as the existence of a common sufficient set for one treatment $A$ and two outcomes $Y_0$ and $Y_1$ is guaranteed by d-separation rules. Intuitively, this is because (i) a causal DAG requires all common causes of two variables be drawn, so that a sufficient set always exists for each outcome considered separately; (ii) we can always consider the union of the two sets; and (iii) since sufficient sets need never include descendants of the treatment \citep{shpitser2012validity}, any variable in one set that harms identification for the other must be a collider (or descendant) on a path containing a common cause of the collider and treatment; such a path can therefore be blocked by a variable not harming identification for either outcome.
\begin{proposition}[Existence of a common minimally sufficient set]\label{prop:existence}
Let $A$ be a treatment and $Y_0$ and $Y_1$ two outcomes that can be represented on a causal DAG. Then there must exist $M$ such that $A \ind Y_1^0|M$ and $A \ind Y_0^0|M$.
\end{proposition}
\begin{proof}
First, note the following \citep{shpitser2012validity}:
\renewcommand{\labelenumi}{(\alph{enumi})}
\begin{enumerate}
    \item If a covariate set $Z$ is sufficient for outcome $Y_t$, then in order for a covariate set $Z'=\{Z, U\}$ to not be sufficient, it must be the case that $U$ contains a collider, a descendant of collider, or a descendant of the treatment
    \item Any minimally sufficient set must not contain any descendants of the treatment
\end{enumerate}

Let $M_0$ be a minimally sufficient set for $Y_0^0$ and let $M_1$ be a minimally sufficient set for $Y_1^0$. Then, suppose that $M=M_0 \cup M_1$ is not a common sufficient set. This implies that one or both of the following is true:
\renewcommand{\labelenumi}{(\roman{enumi})}

\begin{enumerate}
    \item $A \not\ind Y_1^0|M$, or 
    \item $A \not\ind Y_0^0|M$ 
\end{enumerate}
Consider case (i). Since $M_0$ is a minimally sufficient set, from (b), we have that no variable in $M_0$ is a descendant of $A$. From (a), there must exist a variable $H_0\in M_0$ that is a collider or a descendant of a collider on a path between $A$ and $Y_1^0$ that is not otherwise blocked conditional on $M$. Therefore $H_0$ must be connected to $A$ via one or more paths of the form $A\leftarrow \cdots \rightarrow H_0$. Each such path must contain a non-collider between $A$ and $H_0$ such that each path could be blocked by replacing $H_0$ with a non-collider between $A$ and $H_0$. Therefore, there must be a sufficient set for $Y_0$ containing no colliders (or descendants of colliders) on any paths between $Y_1^0$ and $A$. 

Then consider case (ii). From (a) and (b), we have that there must be a variable $H_1\in M_1$ that is a collider or a descendant of a collider on a path between $A$ and $Y_0^0$. Since there are no descendants of $A$ in $M_1$, all paths between $A$ and $H_1$ must contain a non-collider. Therefore, there must be a sufficient set for $Y_1$ containing no colliders or descendants of colliders on any paths between $Y_0^0$ and $A$. 

Let $M_0'$ denote a sufficient set for $Y_0$ not containing any colliders or descendants of colliders between $Y_1^0$ and $A$, and let $M_1'$ denote a sufficient set for $Y_1$ not containing an colliders or descendants of colliders between $Y_0^0$. Then, $M'=M_1'\cup  M_0'$ must be a common sufficient set.  
\end{proof}
\section{Proof of Lemma \ref{lemma:allcommon}}\label{app:allcommon}
\begin{proof}
 Note first that additive homogeneous confounding (\ref{eq:ahc}) implies 
 \begin{equation}\label{eq:subset}
     E(Y_1^0-Y_0^0|M_0)=\E(Y_1^0-Y_0)
 \end{equation}
 for any subset $M_0\subset M$. Also note that (\ref{eq:ahc}) implies 
 \begin{equation}
    \label{eq:ext_ahc}
    \E(Y_1^0-Y_0^0|A,M)=\E(Y_1^0-Y_0^0),
 \end{equation}
 because the sufficiency of $M$ allows us to freely add $A$ to the conditioning event. Then, without loss of generality, let $s=0$ and $s'=1$, and note that 
\begin{align*}
    \E(\dot Y_1^0|A, M_0) &=\E\{\E(\dot Y_1^0|A,M)|A, M_0\} \\
    &=\E\{\E(\dot Y_0^0|A,M)|A, M_0\}\\
    &=\E(\dot Y_0^0|A, M_0)\\
    &=\E( \dot Y_0^0| M_0)\\
    &=\E(\dot Y_1^0| M_0),
\end{align*}
where the second equality applies (\ref{eq:ext_ahc}), the fourth uses the sufficiency of $M_0$ for $Y_0$ (our premise), and the fifth applies (\ref{eq:subset}).
\end{proof}

\section{R code to demonstrate minimally sufficient sets in Figure \ref{fig:swig3} }
\begin{verbatim}
library(dagitty)
figure4_swig = dagitty('pdag {
bb="0,0,1,1"
A [exposure,pos="0.275,0.3"]
"|a=0" [pos="0.29, 0.3"]
U1 [pos="0.25,0.5"]
U2 [pos="0.25,0.2"]
U3 [pos="0.2,0.4"]
U4 [pos="0.35,0.4"]
Y0 [pos="0.2,0.3"]
"Y1^0" [pos="0.35,0.3"]
"|a=0" -> "Y1^0"
U1 -> A; U1 -> U4; U1 -> Y0; U1 -> "Y1^0"
U2 -> Y0; U2 -> "Y1^0"; U2 -> U4
U3 -> A; U3 -> Y0; U3 -> U4
U4 -> A; U4 -> "Y1^0"
Y0 -> U4
U1 -- U2; U1 -- U3; U2 -- U3
}
')
plot(figure4_swig)
adjustmentSets(figure4_swig, outcome='Y1^0', type = 'minimal')
adjustmentSets(figure4_swig, outcome='Y0',   type = 'minimal')
\end{verbatim}
\end{appendices}

\end{document}